\documentclass[lettersize,journal]{IEEEtran}

\usepackage{graphicx}
\usepackage{footnote}
\usepackage{color}
\usepackage{xcolor}
\usepackage{physics}
\usepackage{amsthm}
\usepackage{amsmath, amssymb}
\usepackage{stfloats}
\usepackage[normalem]{ulem}
\usepackage{epstopdf}
\usepackage[compress]{cite}
\usepackage{array}
\usepackage{bm}
\usepackage{booktabs}
\usepackage{braket}
\usepackage{url}
\usepackage[numbers]{natbib}
\usepackage{algorithm}
\usepackage[noend]{algpseudocode}
\usepackage{mathtools}
\usepackage{bigints}
\usepackage{makecell}
\usepackage{float}
\usepackage[caption=false,font=normalsize,labelfont=sf,textfont=sf]{subfig}

\ifCLASSINFOpdf

\else

\fi

\algnewcommand{\Inputs}[1]{%
  \State \textbf{Inputs:}
  \Statex \hspace*{\algorithmicindent}\parbox[t]{.8\linewidth}{\raggedright #1}
}

\makeatletter
\newcommand{\thickhline}{%
    \noalign {\ifnum 0=`}\fi \hrule height 1pt
    \futurelet \reserved@a \@xhline
}
\newcolumntype{"}{@{\hskip\tabcolsep\vrule width 1pt\hskip\tabcolsep}}
\makeatother

\makeatletter
\NewCommandCopy\@@pmod\pmod
\DeclareRobustCommand{\pmod}{\@ifstar\@pmods\@@pmod}
\def\@pmods#1{\mkern4mu({\operator@font mod}\mkern 6mu#1)}
\makeatother

\algnewcommand{\Initialise}[1]{%
  \State \textbf{Initialise:}
  \Statex \hspace*{\algorithmicindent}\parbox[t]{.8\linewidth}{\raggedright #1}
}

\newtheorem{theorem}{Theorem}

\newtheoremstyle{named}{}{}{\itshape}{}{\bfseries}{.}{.5em}{\thmnote{#3's }#1}
\theoremstyle{named}

\begin{document}
{\renewcommand{\arraystretch}{1.2}
\newcommand{\listelist}[2]{
	\begin{enumerate}
	\item #1
	\item #2
	\end{enumerate}}

\title{ Modelling 1D Partially Absorbing Boundaries for Brownian Molecular Communication Channels\\[.1ex]}

\author{Caglar~Koca,~\IEEEmembership{Student Member,~IEEE,} and Ozgur~B. Akan,~\IEEEmembership{Fellow,~IEEE}
\thanks{Caglar Koca and Ozgur B. Akan are with the Internet of Everything (IoE) group, University of Cambridge, CB3 0FA, Cambridge, UK (email: \{ck542,oba21\}@cam.ac.uk).}
\thanks{Ozgur B. Akan is also with the Center for neXtGeneration Communications (CXC), Department of Electrical and Electronics
Engineering, Koç University, 34450 Istanbul, Turkey (email: akan@ku.edu.tr).}
\thanks{This work was supported by AXA Research Fund (AXA Chair for Internet of Everything at Koç University).}
}

\maketitle

\begin{abstract}
Molecular Communication (MC) architectures suffer from molecular build-up in the channel if they do not have appropriate reuptake mechanisms. The molecular build-up either leads to intersymbol interference (ISI) or reduces the transmission rate. To measure the molecular build-up, we derive analytic expressions for the incidence rate and absorption rate for one-dimensional MC channels where molecular dispersion obeys the Brownian Motion. We verify each of our key results with Monte Carlo simulations. Our results contribute to the development of more complicated models and analytic expressions to measure the molecular build-up and the impact of ISI in MC.
\end{abstract}
\begin{IEEEkeywords}
Molecular Communications, intersymbol interference, molecular reuptake, channel clearance.
\end{IEEEkeywords}

\section{Introduction}
\label{sec:intro}

Molecular communication (MC) is the primary mode of communication in the nature. It is prevalent both among organisms, in the forms of pollen, spores and pheromones and within organisms in the forms of neurotransmitters, enzymes and hormones \cite{akan2016fundamentals, vizziello2023intra, castorina2016analytical}. However, unlike the electromagnetic waves, the information carrying molecules (IcM) tend to accumulate in the medium. Thus, MC architectures, especially when used in a confined volume, need a {\it reuptake} mechanism to prevent accumulation. In this work, we pave the way for building complex models describing molecular reuptake.

In an MC channel, the IcM build-up in the medium either dramatically reduces the transmission rate or leads to inter-symbol interference (ISI) \cite{tepekule2015isi, chen2017modeling, hyun2023isi, aktas2023odor, bilgen2024odor}. In nature, healthy organisms have well functioning reuptake mechanisms to mitigate ISI, which is one of the bottlenecks of MC \cite{dong2009molecular,ramezani2017rate}. Realising this, many prominent works include an absorbing receiver in their envisioned MC architectures to reduce the time to prepare the medium for the next transmission \cite{dissanayake2017reed, bao2021relative, kadloor2012molecular, vakilipoor2023localizing,yaylali2023channel,huang2023physical}. However, in SMC, absorption does not occur at the receivers but on the dedicated reuptake molecules. Thus, these works are not directly applicable to SMC \cite{pajarillo2019role, divito2014excitatory}.

Apart from the absorbing receiver assumption, most of the existing works in the MC literature generally use either simple mechanisms to describe the reuptake or ignore it altogether. The simple mechanisms include constant ratio reuptake on the boundaries and infinite volume assumption \cite{khan2017diffusion}.  Infinite volume assumption unnecessitates any reuptake mechanism as IcM can disperse to infinity. The system models that do not include any reuptake mechanism implicitly use the infinite volume assumption, allowing the IcM to disperse to and from the infinity \cite{noel2013using, sajjad2023high, wang2020understanding}.

More realistic models solve a boundary value problem (BVP). Depending on the system model, reuptake may occur at the receiver \cite{yilmaz2014three} or at the boundaries of the volume \cite{koca2022channel, lotter2020channel}. Solving the BVPs are easier for perfectly absorbing boundaries, as they only use the Dirichlet boundary condition, whereas partially absorbing boundaries use a combination of Dirichlet and Von Neumann boundary conditions. There are works in the literature that tackle this issue for spherical boundaries \cite{arjmandi2019diffusive, zoofaghari2021modeling} or works which might be extended to include the reuptake process \cite{lotter2020synaptic}.

In cases where BVPs are too hard to solve, Monte Carlo simulations are frequently used to approximate a solution \cite{koca2022narrow, koca2021molecular}. However, when we use Monte Carlo simulations to model partially absorbing boundaries further complications arise, as the rate of absorption depends on the space step used in the system \cite{erban2007reactive, koca2023fast}. If the absorption rate to the diffusion coefficient ratio is too high, then the space step might be too small to complete the simulation in reasonable time. Alternatively, a large space step would necessitate a large absorption probability, which might even exceed 1.

To address this issue, we derive analytic expressions for the second incidence rate in 1D. While our primary motivation is to find the absorption rate in MC, we believe second incidence rates might be of interest to the researchers of other disciplines, such as finance. Then, using the second incidence rates, we reach an analytic result for the absorption rate of the IcM. Although we work in a one-dimensional domain, our results can be extended to higher dimensional systems.

Our itemised major contributions in this work are
\begin{itemize}
\item deriving an expression for the second incidence rate in a semi-infinite 1D medium,
\item deriving an expression for the second incidence rate in a bounded 1D medium,
\item deriving an expression for the absorption rate in a bounded 1D medium.
\end{itemize}

The rest of this paper is organised as follows. In Sec. \ref{sec:sys}, we present the background information and the system model. In Sec. \ref{sec:no_reflect}, we find the second incidence rate in a discretised, one-dimensional semi-infinite medium. In Sec. \ref{sec:reflect_cont}, we derive an approximate expression for the second incidence rate in a medium bounded at both ends. Then, in Sec. \ref{sec:reflect}, we obtain an exact analytical expression for the second incidence rate in a bounded volume. We reach an analytical expression for the absorption rate in Sec. \ref{sec:rateabsicm}. Sec. \ref{sec:concpap01} is the conclusion.

\section{System Model}
\label{sec:sys}

In this work, assume a Brownian Motion, $\mathcal{B}$, in a one-dimensional domain with partially reflective end point(s). Except for Sec. \ref{sec:no_reflect}, where we use a semi-infinite domain with only one reflective end point, we always assume a finite domain of length $2L$. We discretise this domain into $2H$ segments, resulting in 
\begin{align}
\label{eq:L_def}
\Delta{x}=\frac{H}{L}.
\end{align}

After discretisation, we choose nodes as the centre poins of these $2H$ segments, i.e., located at $[-H+0.5, -H+1.5,\dots -0.5, 0.5, \dots, H-0.5]$, where we assume that $x=-H+0.5$ is the leftmost node of the domain. Apart from Sec. \ref{sec:no_reflect}, where the only boundary is located at $x=H$, both boundaries at $x\pm H$ are partially reflective.

We use the fixed jump model for $\mathcal{B}$, i.e., when $\mathcal{B}[N-1] \notin \{-H+0.5, H-0.5\}$, the movement is described by
\begin{equation}
\label{eq:brown_mot_def}
\mathcal{B}[N]=\begin{cases}&\mathcal{B}[N-1]+1, \hspace{1mm} \text{w.p. 0.5}\\
&\mathcal{B}[N-1]-1, \hspace{1mm} \text{w.p. 0.5}
\end{cases}.
\end{equation}

\eqref{eq:brown_mot_def} dictates that the time step and the space step are related to each other through the diffusion coefficient, $D$, i.e.,
\begin{equation}
\label{eq:D_def}
D=\frac{\Delta{x}^2}{2\Delta{t}}.
\end{equation}

At the partially reflected boundary, $\mathcal{B}$ is either absorbed with a probability $P_A$ or reflected. If reflected, $\mathcal{B}$ remains at its previous position, i.e.,
\begin{align}
\mathcal{B}[N] = \mathcal{B}[N-1].
\end{align}

While we use the fixed jump model, there are other models. According to the variable jump model, $\mathcal{B}[N]$ is 
\begin{equation}
\label{eq:alt_brown_def}
\mathcal{B}[N] = \mathcal{B}[N-1] + \delta\sqrt{2D\Delta{t}},
\end{equation}
where $\delta$ is a random variable obeying standard normal distribution \cite{erban2007reactive}. We refer to this model in Sec. \ref{sec:reflect_cont} and Sec. \ref{sec:rateabsicm} but use the fixed jump model consistently in the entire work.

Our focus in this work is to model the second incidence rate. Accordingly, unless otherwise stated, we assume that the first incidence occurs at $t=0$. We use the rate of first incidence as derived in \cite{koca2022channel} when needed.

Throughout this work, we use Monte-Carlo simulations to support the theoretical results we present. Table \ref{tab:sim_param} provides the default values of the independent simulation parameters and the relations with which the dependent parameters are derived. We use the default values listed in Table \ref{tab:sim_param} in our simulations, unless we state otherwise. To assist the reader, we also list the recurring notations in Table \ref{tab:text_sym}. While we define each symbol in text when they are first referred to, we believe a reference table could be beneficial to the reader. 

\begin{table}[t]
\caption {Base parameters for simulations}
\centering
\begin{tabular}{l|ccc}
\textbf{Variable}              & \textbf{Value}        & \textbf{Unit} & \textbf{Relation} \\ \thickhline
\textbf{Diffusion Coefficient}\hspace{6.0mm} ($D$)              & $1$  & $\text{m}^{2}\text{s}^{-1}$ & \\
\textbf{Domain Length}\hspace{11.9mm} ($2L$)                       & $4$    & $\text{m}$     &     \\
\textbf{Space Step}\hspace{17.7mm} ($\Delta{x}$)                 & $0.05$ & $\text{m}$   & \\
\textbf{Absorption Rate}\hspace{10.8mm} ($\lambda$)                 & $0.1$ & $\text{m}\text{s}^{-1}$   & \\
\textbf{Space Discretisation}\hspace{6.55mm} ($H$)                 & &    & \eqref{eq:L_def}\\
\textbf{Time Step}\hspace{18.5mm} ($\Delta{t}$)                 & & $\text{s}$   & \eqref{eq:D_def}\\
\textbf{Simulation Time}\hspace{10.7mm} ($t$)                 & & $\text{s}$   & \\
\textbf{Time Discretisation}\hspace{7.3mm} ($N$)                 & &    & \eqref{eq:N_def}\\
\textbf{Absorption Probability}\hspace{2.85mm} ($P_A$)                 & &    & \eqref{eq:erban_7}\\

\end{tabular}
\label{tab:sim_param}
\end{table}

\begin{table}[t]
\caption {List of Notations for Recurring Functions}
\centering
\setlength{\tabcolsep}{8pt}
\begin{tabular}{ll}
\textbf{Description} & \textbf{Symbol}\\ \thickhline
\textbf{Probabilities in semi-infinite domain}& \\
\hspace{4mm} {Absorption Prob. at $N$ with $\mathcal{B}[0]=H$} & $\mathrm{P^{(S)}}[M=N]$  \\
\hspace{4mm} {Survival Prob. at $N$ with $\mathcal{B}[0]=H$} & $\mathrm{P^{(S)}}[M>N]$  \\
\hspace{4mm} {Survival Prob. at N with $\mathcal{B}[0]=H-x^\prime$}& $\mathrm{P_{x^\prime}^{(S)}}[M>N]$  \\
\textbf{Probabilities in finite domain} & \\
\hspace{4mm} {Survival Prob. at $N$ with $\mathcal{B}[0]=H$} & $\mathrm{P}[M>N]$  \\
\hspace{4mm} {Survival Prob. at $N$ with $\mathcal{B}[0]=H-x^\prime$}& $\mathrm{P_{x^\prime}^{(S)}}[M>N]$ \\
\textbf{Rates in Discrete Time} & \\
\hspace{4mm} {Rate of  Incidence} & $R[N]$               \\
\hspace{4mm} {Rate of {1\textsuperscript{st}}  Incidence with $\mathcal{B}[0]=0$}& $R_1[N]$  \\
\hspace{4mm} {Rate of {2\textsuperscript{nd}} Incidence with $\mathcal{B}[0]=H$}& $R_2[N]$  \\
\hspace{4mm} {Rate of {2\textsuperscript{nd}}  Incidence with $\mathcal{B}[0]=x^\prime$}& $R_2[x^\prime,N]$  \\
\hspace{4mm} {Rate of Higher Order Incidences with $\mathcal{B}[0]=H$} & $R_h[N]$   \\
\hspace{4mm} {Rate of Absorption}& $R_A[N]$  \\
\end{tabular}
\label{tab:text_sym}
\end{table}

\section{Second Incidence Rate with one Boundary}
\label{sec:no_reflect}

In this section, we start with a semi-infinite domain with a partially reflective boundary at $x=H$, i.e., $\mathcal{B} \in (-\infty, H]$. By definition, to find the rate of second incidence, we assume that $\mathcal{B}$ was incident at $x=H$. Thus, without loss of generality, we can shift $\mathcal{B}$ in time steps such that $\mathcal{B}[0] = H-0.5$.  In this section, since we only focus on the second incidence rate, we choose the probability of absorption, $P_A$ as 1. 

Assuming the second incidence happens at time step $M$, we calculate $\mathrm{P^{(S)}}[M>N]$, the probability that the second incidence has not yet occurred in the next $N$ time steps, where
\begin{align}
\label{eq:N_def}
N=\frac{t}{\Delta{t}}.
\end{align}

Let $\mathcal{B}$ take $n_l$ steps to the $-x$ direction and $n_r$ steps to the $+x$ direction, satisfying $n_r + n_l = N$. Since we know that for $n_r > n_l$, $\mathcal{B}$ necessarily hits $x=H$ at least once within the first $N$ time steps. Thus, we choose $n_l \geq n_r$. Using Bertrand's Ballot Theorem with ties allowed, we reach
\begin{align}
\label{eq:bertnard}
\mathrm{P^{(S)}}[M>n_r+n_l|n_r,n_l]=\frac{n_l-n_r + 1}{n_l + 1}.
\end{align}

Since ties are allowed, \eqref{eq:bertnard} includes cases where $\mathcal{B}$ returns back to $x=H-0.5$ but $\mathcal{B}$ never exceeds $H$. Moreover, the incidences occur only at odd values of $N$, due to the $\mathcal{B}[0]=H-0.5$ initial condition.

Since there are $2^N$ sequences of $N$-long $+x$ and $-x$ movements, the probability that $\mathcal{B}$ is never incident on the boundary within the first $N$ time steps is
\begin{align}
\label{eq:bertrand_sum}
\mathrm{P^{(S)}}[M>N] &=\frac{1}{2^N}\sum_{k=n}^{N}\binom{N}{k}\frac{2k-N+1}{k+1},
\end{align}
where 
\begin{equation}
n = \left\lceil \frac{N}{2} \right\rceil.
\end{equation}

The binomial term in \eqref{eq:bertrand_sum} provides the number of sequences with $k$ moves in the $+x$ direction and the fractional term gives the probability that a given sequence with $k$ moves in the $+x$ direction satisfies \eqref{eq:bertnard}. The choice of limits for the summation reflects that $k$, the number of steps taken in the $+x$ direction, is larger than $N-k$, the number of steps in the $-x$ direction.

Starting with \eqref{eq:bertrand_sum}, we continue as
\begin{align}
\mathrm{P^{(S)}}[M>N] &= \frac{1}{2^{N}}\sum_{k=n}^{N}\binom{N}{k}\left(2-\frac{N+1}{k+1}\right),\\
\label{eq:bertrand_1}
&= \frac{1}{2^{N}}\sum_{k=n}^{N}2\binom{N}{k}-\binom{N+1}{k+1},\\
\label{eq:bertrand_exact}
&= \frac{1}{2^{2n}}\binom{2n}{n}.
\end{align}

Now, we can use \eqref{eq:bertrand_exact} to find $\mathrm{P^{(S)}}[M=N]$, i.e.,
\begin{align}
\mathrm{P^{(S)}}[M=N]&=\mathrm{P^{(S)}}[M>N-2]-\mathrm{P^{(S)}}[M>N],\\
&=\frac{1}{2^{N-2}}\binom{N-1}{n-1}-\frac{1}{2^{N}}\binom{N+1}{n},\\
\label{eq:binom_der}
&=\frac{1}{2^{N-2}}\binom{N-2}{n-1}\frac{1}{N},
\end{align}
for odd $N$ and $\mathrm{P^{(S)}}[M=N]=0$ for even $N$. 

Using the Stirling approximation of the form 
\begin{align}
\label{eq:pr_full_log2}
\log(N!)=\frac{1}{2}\ln(2\pi N)+N\ln(N)-N,
\end{align}
\eqref{eq:bertrand_exact} turns into
\begin{align}
\mathrm{P^{(S)}}[M>2n]&\approx \frac{1}{2^{2n}}\frac{\sqrt{2\pi2 {n}}\left(\frac{2n}{e}\right)^{2n}}{\left(\sqrt{2\pi n}\left(\frac{n}{e}\right)^{n}\right)^2},\\
\label{eq:bertrand_pareto}
&\approx \sqrt{\frac{1}{\pi n}}.
\end{align}

Similarly, \eqref{eq:binom_der} becomes
\begin{align}
\mathrm{P^{(S)}}[M=2n+1]&\approx \frac{1}{2n\sqrt{\pi (n-1)}},\\
\label{eq:bertrand_pareto1}
&\approx \sqrt{\frac{1}{4\pi n^3}}.
\end{align}

Note that both \eqref{eq:bertrand_pareto} and \eqref{eq:bertrand_pareto1} are as good as the Stirling approximation, i.e., they are inaccurate for small $m$. Thus, for smaller values, it might be better to use \eqref{eq:bertrand_exact} and \eqref{eq:binom_der} rather than their approximations.

\begin{figure}[t]
	\centering	
	\includegraphics[width=3.5in, trim={0cm 0 0cm 0}]{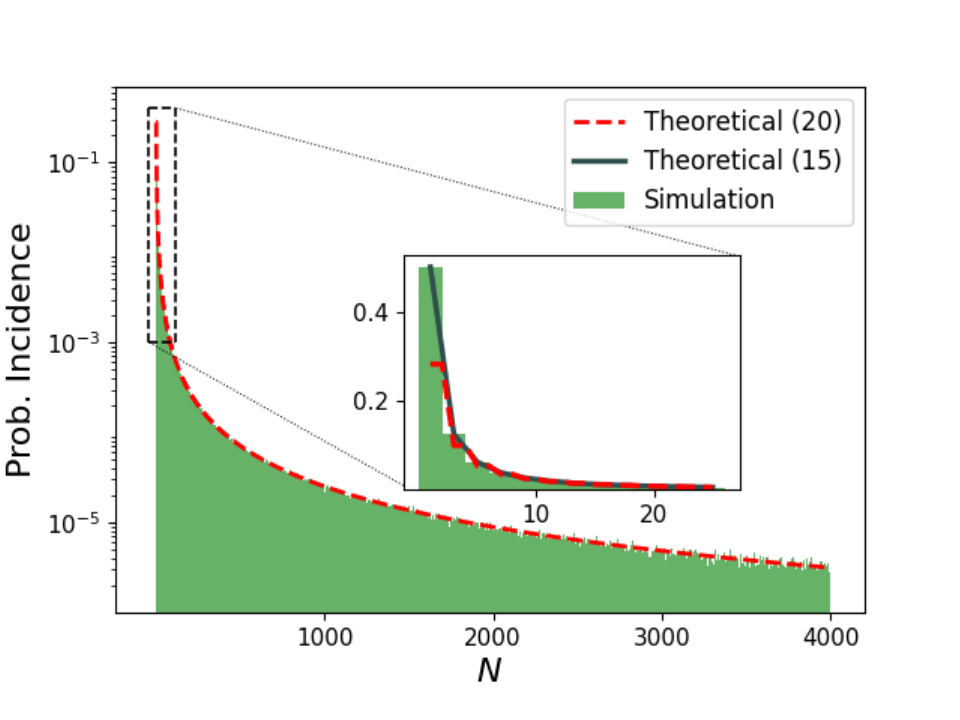}
	\vspace{-5mm}
	\caption{Probability of hitting at time step $N$ vs. time steps following the first hit. We employ semi-log plot for the outer figure, to increase the visibility of the data for large $N$, while the zoomed in plot is linear. We iterate the simulation for 20000000 times to obtain the simulation results. To simplify the plot, we only plot the simulation results only for odd $N$, i.e., we do not show the null hit probability at even $N$. Similarly, we choose $\mathrm{P^{(S)}}[N] = \mathrm{P^{(S)}}[N-1]$ for even $N$ to prevent zigzagging of the theoretical results.}
	\label{fig:eq20}
\end{figure}

Fig. \ref{fig:eq20} illustrates the agreement between the simulations and \eqref{eq:bertrand_pareto1}. As expected, for small $N$ for which the Stirling approximation is inaccurate, \eqref{eq:bertrand_pareto1} does not return particularly accurate results. However, for large $N$, our theory is in unison with the simulation results. Also evident from Fig. \ref{fig:eq20}, \eqref{eq:binom_der} gives the exact incidence probability. Since binomial terms grow exponentially, \eqref{eq:binom_der} has limited use. Regardless, we can use both \eqref{eq:binom_der} and \eqref{eq:bertrand_pareto1} for small and large $N$ respectively to get an accurate and robust result.

Note that the second incidence probability is given by a Pareto distribution with a scale parameter $x_m=(1/2\pi)$ and a shape parameter $\alpha=0.5$, where \eqref{eq:bertrand_pareto} and \eqref{eq:bertrand_pareto1} are the the cumulative and probability distribution functions respectively. Pareto distribution with such parameters is heavy-tailed, i.e., the expectation value of the second incidence time step is infinite. Although seemingly counter-intuitive, this is expected, as $\mathcal{B}$ makes larger and larger swings into $-\infty$ direction as long as $\mathcal{B}$ survives.

\section{Time to Second Incidence in a Finite Domain}
\label{sec:reflect_cont}

Due to the expected second hitting time being infinity, the results that we obtained in Sec. \ref{sec:no_reflect} does not offer much help in modelling realistic systems, as $\mathcal{B}$ traversing to and from either $-\infty$ or $\infty$ is not realistic. In finance, we know that price of a stock would crash at zero. In MC, IcM would reflect from the sides of the {\it finite} container it is in.

To find the second incidence rate for a realistic system, as we describe in Sec. \ref{sec:sys}, we move to a finite system with boundaries located at $x=\pm H$. Since we focus on the second incidence, we assume that the first incidence already happened and $\mathcal{B}$ is now located at $L-\Delta{x}$. We also choose the boundary at $x=\pm H$ are absorbing as  the second incidence occurs when $\mathcal{B}$ is incident on $x=\pm H$, and any further incidences are no longer second incidence but higher order incidences.

Note that we choose $L-\Delta{x}$ instead of $L-\Delta{x}/2$, as described in Sec. \ref{sec:sys}. The reasoning behind our choice is that simulations cannot capture half time steps, i.e., starting at $L-\Delta{x}/2$, $\mathcal{B}$ is incidence at $x=H$ at $\Delta{t}/2$ with a probability of $0.5$. However, the computer simulations cannot capture the half time steps. Thus, we altered the $\mathcal{B}[0]$ to fit the simulation results to the theory.

\begin{figure}[!b]
	\centering
	\includegraphics[width=3.6in, trim={0.0cm 0 0.0cm 0}]{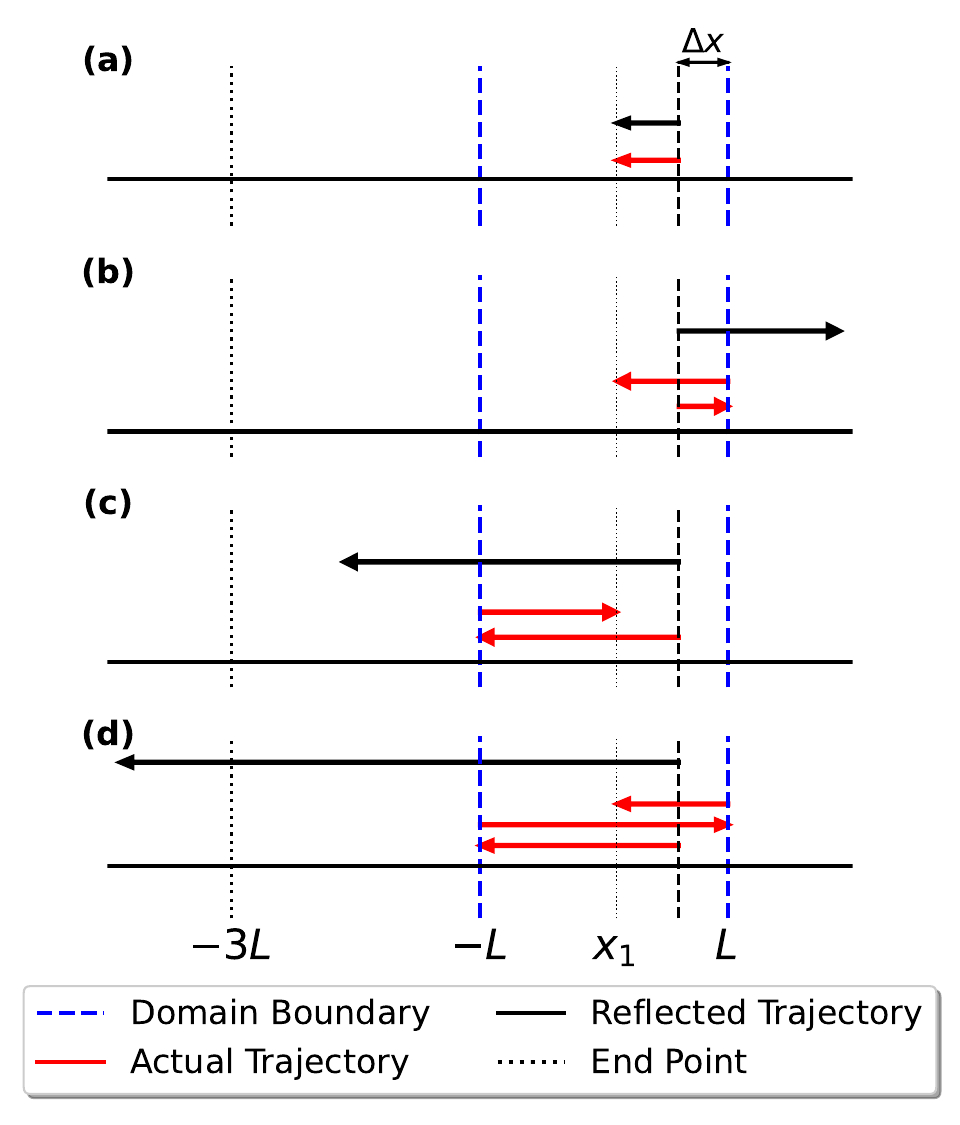}
	\caption{Illustration of the reflection principle, as outlined in  \cite{koca2022channel}. Each subfigure shows one of the countably infinite paths $\mathcal{B}$ can use, starting at $x=L-\Delta{x}$ and ending at $x=x_1$.}
	\label{fig:transformed}
\end{figure}

Using the start point and the end point of $\mathcal{B}$ does not divulge enough information whether $\mathcal{B}$ were incident on the boundaries. It may return back to the domain after it crossed the boundary at one, or many, incidences. To account for the crossings, we use the reflection principle \cite{koca2022channel}. We illustrate the reflection principle in Fig. \ref{fig:transformed}. Among the four possible paths $\mathcal{B}$ might take, only the path outlined in Fig. \ref{fig:transformed}(a) does not cross any of the boundaries. In Fig. \ref{fig:transformed}(b), $\mathcal{B}$ crosses the boundary at $x=H$, in (c) at $x=-H$, and in (d) first at $x = -H$ and at $x=H$, even though they all conclude at the same point in space.

To discount paths involving boundary crossings, we adjust (11) of \cite{koca2022channel} as described in Sec. 4(b) of the same work. The probability that the time steps for the next hit is smaller than $M$ is 
\begin{align}
\label{eq:siegmund5}
\begin{split}
\mathrm{P}&[N<M]=1-\sum_{i=-\infty}^\infty(-1)^{i}\times\\ &\phantom{ooo}\left[\Phi\left(\frac{2iH +1}{\sqrt{N}}\right)-\Phi\left(\frac{(2i-2)H + 1}{\sqrt{N}}\right)\right],
\end{split}
\end{align}
where $\Phi(.)$ is the standard normal cumulative distribution function.

Assuming $H\gg 1$, we use Taylor expansion on \eqref{eq:siegmund5}, i.e.,
\begin{align}
\nonumber
\mathrm{P}&[N<M]\!=\!1-\sum_{i=-\infty}^\infty(-1)^{i}\!\times \! \left[\Phi\left(\frac{2iH}{\sqrt{N}}\right)\!+\!\frac{1}{\sqrt{N}}\phi\left(\frac{2iH}{\sqrt{N}}\right)\right.\\
\label{eq:siegmund6} 
&\phantom{oooo}\left.-\Phi\left(\frac{(2i-2)H}{\sqrt{N}}\right)-\frac{1}{\sqrt{N}}\phi\left(\frac{(2i-2)H}{\sqrt{N}}\right)\right],
\end{align}
where $\phi(.)$ is the standard normal probability distribution function.

We first manipulate the $\Phi(.)$ terms of \eqref{eq:siegmund6} as,
\begin{align}
\label{eq:non_div0}
Z &=\! \sum_{i=-\infty}^\infty(-1)^{i}\times \left[\Phi\left(\frac{2iH}{\sqrt{N}}\right)-\Phi\left(\frac{(2i-2)H}{\sqrt{N}}\right)\right],\\
\label{eq:non_div}
&= \lim_{Q\rightarrow \infty} 2\sum_{i=-Q}^Q  (-1)^{i}  \Phi\left(\frac{2iH}{\sqrt{N}}\right), \\
\label{eq:non_div2}
& = \lim_{Q\rightarrow \infty}1\!+\!2\!\sum_{i=1}^Q (-1)^{i}\!\left[\Phi\left(\frac{2iH}{\sqrt{N}}\right)\!+\!\Phi\left(\frac{-2iH}{\sqrt{N}}\right)\right].
\end{align}

Since we know that $\Phi(\alpha) + \Phi(-\alpha) =1$, \eqref{eq:non_div0} clearly diverges. However, we realise that $\Phi(.)$ terms vanish as long as the limits of integration are symmetric. Thus, we rearrange \eqref{eq:non_div} to \eqref{eq:non_div2} where each term sums to zero, allowing us to calculate the Cesaro sum of the series as zero. After $\Phi(.)$ terms vanish, \eqref{eq:siegmund6} becomes
\begin{align}
\label{eq:siegmund7} 
\mathrm{P}[N<M]&=\frac{1}{\sqrt{N}}\sum_{i=-\infty}^\infty(-1)^{i}\left[\phi\left(\frac{2iH}{\sqrt{N}}\right)\!+\!\phi\left(\frac{-2iH}{\sqrt{N}}\right)\right],\\
&=\frac{2\phi(0)}{\sqrt{N}} + \frac{4}{\sqrt{N}}\sum_{i=1}^{\infty} (-1)^i \phi\left(\frac{2iH}{\sqrt{N}}\right),\\
\label{eq:vartheta0} 
&=\sqrt{\frac{2}{\pi N}}\vartheta_4(0;q),
\end{align}
where 
\begin{equation}
\label{eq:q_def}
q = \exp\left(-\frac{2H^2}{N}\right)
\end{equation}
and $\vartheta_4(z;q)$ is the fourth Jacobi Theta function of the form
\begin{equation}
\label{eq:var_def}
\vartheta_4(z;q) = 1 + 2\sum_{i=1}^\infty (-q)^{i^2}\cos(2iz).
\end{equation}

As we stated in Sec. \ref{sec:no_reflect}, $\mathcal{B}$ cannot be incident on the boundary at even $N$, however \eqref{eq:vartheta0} returns a positive value for both even and odd $N$. Thus, to ensure that $\mathrm{Pr}[M>N]=\mathrm{Pr}[M>N-1]$ for even $N$, we modify \eqref{eq:vartheta0} as
\begin{align}
\label{eq:vartheta00} 
\mathrm{P}[N<M]=
\sqrt{\frac{2}{\pi (N+\eta)}}\vartheta_4\left(0;\exp\left(\frac{-2H^2}{N+\eta}\right)\right),
\end{align}
where 
\begin{align}
\eta \equiv (N+1) \hspace{-2mm}\pmod{2}.
\end{align}

With Fig. \ref{fig:eq30}, we present the parity of the simulation results and \eqref{eq:vartheta00}. The discrepancy for small $N$, clearly visible in Fig. \ref{fig:eq30}(a) is stemmed from the fact that the starting point for \eqref{eq:vartheta00} uses the Brownian Motion defined with \eqref{eq:alt_brown_def}, instead of \eqref{eq:brown_mot_def}. However, for large $N$, both models converge as expected. Thus, \eqref{eq:vartheta00} provides a very accurate picture to the probability of not hitting the boundary.

\begin{figure}[t]
	\centering	
	\includegraphics[width=3.8in, trim={1.2cm 0 0cm 0}]{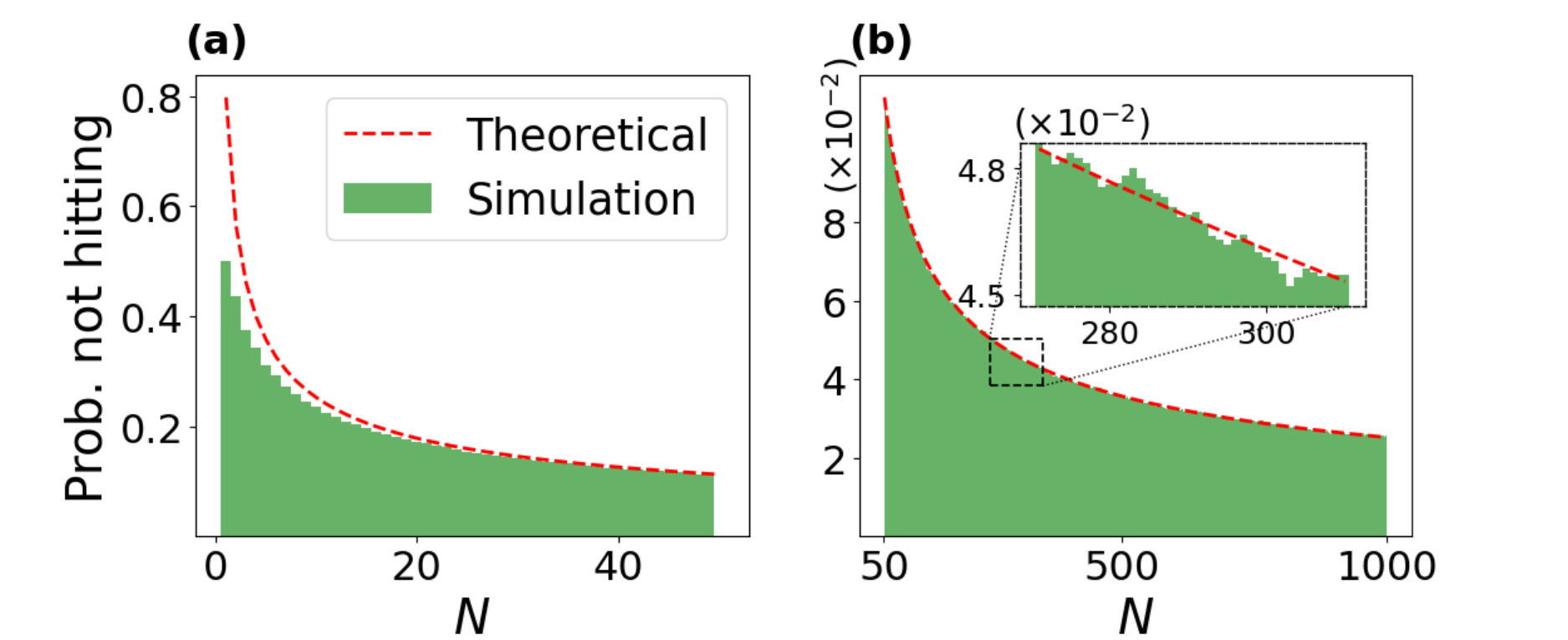}
	\vspace{-5mm}
	\caption{The probability of not hitting vs. number of timesteps following the first hit for $H=200$. We use 500000 iterations to obtain the simulation results. Although \eqref{eq:vartheta00} provides values for even $N$, it is impossible to hit the boundary at even timesteps. Thus, we only plotted the simulation results of odd $N$.}
	\label{fig:eq30}
\end{figure}

We find the rate of second incidence as the derivative of \eqref{eq:vartheta00}, i.e.,
\begin{align}
\label{eq:der_start}
R_2[N] &= - \frac{\mathrm{d}}{\mathrm{d}N} \mathrm{Pr}[N<M], \\
&=\frac{1}{\sqrt{2\pi N^3}}\vartheta_4(0;q)+\sqrt{\frac{2}{\pi N}}\frac{\mathrm{d}q}{\mathrm{d}N}\frac{\mathrm{d}\vartheta_4(0;q)}{\mathrm{d}q},\\
\label{eq:rate0}
&=\!\frac{1}{\sqrt{2\pi N^3}}\left[-\vartheta_4(0;q)\!+\!\frac{4H^2}{N}\sum_{i=1}^\infty (-1)^{i}i^2q^{i^2}\right].
\end{align}

Although \eqref{eq:rate0} does not have a closed expression, we know that
\begin{align}
\frac{\mathrm{d}}{\mathrm{d}q}\vartheta_4(0;q) = -\gamma\frac{1}{q} \frac{\mathrm{d^2}}{\mathrm{d}z^2}\vartheta_4(z;q)\Bigg|_{z=0},
\end{align}
which is evaluated very efficiently by computers. Note that $\gamma$ is a constant that depends on the exact definition of $\vartheta_4(q;z)$. With the definition in \eqref{eq:var_def}, $\gamma = 1/4$, however if the summands are in the form $(-q)^{i^2}\cos(2\pi iz)$, then $\gamma = 1/(4\pi^2)$. 

\section{Corrections for the Second Incidence Rate with Two Boundaries for Small $N$}
\label{sec:reflect}

As we discuss in Sec. \ref{sec:reflect_cont}, the accuracy of \eqref{eq:rate0} is limited to large $N$. In this section, we solve the second incidence problem for small $N$.

In this section, we find $\mathrm{P^{(S)}_{x^\prime}}$, the probability of $\mathcal{B}$ not hitting the boundary in a semi-infinite domain, if $\mathcal{B}[0]=H-x^\prime$. In other words, we extend  \eqref{eq:bertrand_exact} to the cases where $\mathcal{B}$ is not adjacent to the boundary at the start. To this end, we first find a recursive relation and then a closed-form expression. Due to the size of their proofs, we state them as theorems and leave the proofs to the Appendix.

\begin{theorem}
\label{thm:recurse}
{\bf (Extended Ballot Theorem)} Let $\mathcal{B}$ be a Brownian Motion with $\mathcal{B}[0]=H-x^\prime$. The probability that $\mathcal{B}$ does not hit $x=H$ after $N$ time steps can be found with the recursive relation
\begin{align}
\hspace{-1mm}\label{eq:process_recursive}
\mathrm{P^{(S)}_{x^\prime}}[M>N] \!&=\! 2\mathrm{P^{(S)}_{x^\prime-1}}[M>N+1]\! -\! \mathrm{P^{(S)}_{x^\prime-2}}[M>N].
\end{align}
\end{theorem}

While \eqref{eq:process_recursive} is sufficient to find $\mathrm{P^{(S)}_{x^\prime}}[M>N]$, it is not computationally efficient for large $x^\prime$. When $\mathcal{B}[0]$ is very close to one boundary, it is far from the other one. In other words, small $x^\prime$ implies a large $2H-x^\prime$, which necessitates a computationally efficient method. Hence, we derive a closed form expression for \eqref{eq:process_recursive} with Theorem \ref{thm:hypergeom}.

\begin{theorem}
\label{thm:hypergeom}
The probability that $\mathcal{B}$ with $\mathcal{B}[0] =H-x$ does not hit $x=H$ after $N$ time steps is
\begin{align}
\label{eq:hypergeom}
\mathrm{P^{(S)}_{x^\prime}}&[M>N] =\frac{2^{\xi^\prime}(\xi+1)^{\xi^\prime}(-1)^\xi}{\sqrt{\pi}}\frac{\Gamma\left(\Delta+\frac{1}{2}\right)}{\Gamma\left(\Delta+1\right)} \times \\
\nonumber
&\phantom{ooo}{}_3F_2\left(-\xi,\xi+\xi^\prime+1,\Delta + \frac{1}{2};\frac{1}{2}+\xi^\prime,\Delta+1;1\right),
\end{align}
where 
$\Gamma$ is the Gamma function and $_3F_2$ is the hypergeometric function of the form
\begin{align}
\label{eq:hypergeom_def}
_3F_2(a_1,a_2,a_3;b_1,b_2;f)=\sum_{i=0}^\infty \frac{(a_1)_i(a_3)_i(a_3)_i}{(b_1)_i(b_2)_i}\frac{f^i}{i!},
\end{align}
with $(a)_i$ being the Pochhammer symbol, i.e.,
\begin{align}
(a)_i = \begin{cases} 1, &\mathtt{if }\phantom{o} i=0,\\
\prod\limits_{j=0}^{i-1}(a+j), &\mathtt{if}\phantom{o} i>0
\end{cases},
\end{align}
and
\begin{align}
\label{eq:def_xi1}
\xi &= \left\lceil \frac{{x^\prime}}{2}\right\rceil - 1,\\
\label{eq:def_xiprime}
\xi^\prime &\equiv {x^\prime}+1 \pmod{2},\\
\label{eq:def_xipprime}
\xi^{\prime\prime} &\equiv N + 1 \pmod{2},\\
\label{eq:def_xippprime}
\Delta &= \frac{N+\xi^{\prime}(1+\xi^{\prime\prime})}{2}.
\end{align}
\end{theorem}

We display the accuracy of \eqref{eq:hypergeom} with Fig. \ref{fig:eq46}. Contrary to \eqref{eq:vartheta00}, \eqref{eq:hypergeom} is not bound by any assumptions, making it accurate for the entire domain. However, for very large $N$, $_3F_2$ becomes computationally inefficient. Thus, for large $N$, we can use \eqref{eq:vartheta00} and \eqref{eq:rate0}, which are faster and without any noticeable inaccuracy as we discuss in Sec. \ref{sec:reflect_cont}.

Recalling our parity discussion on \eqref{eq:vartheta00} presented in Sec. \ref{sec:reflect_cont}, we stress that the parity is already built in \eqref{eq:hypergeom} via \eqref{eq:def_xiprime} and \eqref{eq:def_xipprime}. As a result, \eqref{eq:hypergeom} and \eqref{eq:rate0} do not allow incidence on the boundary if the even/odd parity of $x^\prime$ and $N$ are not equal. Thus, we do not need to further modify \eqref{eq:hypergeom} as we modify \eqref{eq:vartheta00}. Moreover, \eqref{eq:hypergeom} is computationally efficient compared to \eqref{eq:process_recursive}, making it useful for fast simulations.

\begin{figure}[t]
	\centering	
	\includegraphics[width=3.5in, trim={0cm 0 0cm 1.75cm}]{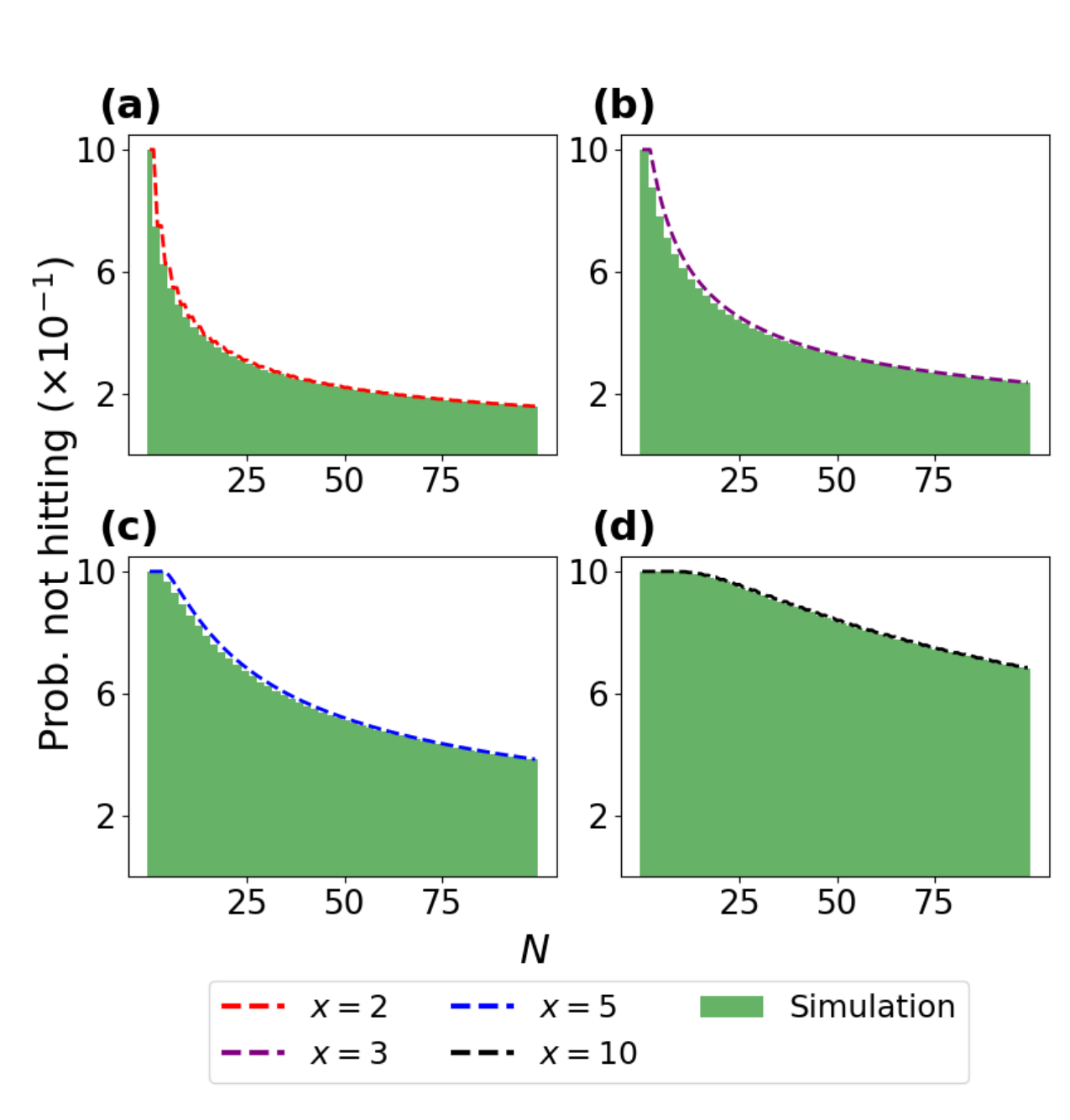}
	\vspace{-6mm}
	\caption{The probability of not hitting the boundary vs. number of timesteps assuming $\mathcal{B}$ starts $x$ distance from the boundary. We use 1000000 iterations each to obtain the simulation results. Similar to  Fig. \ref{fig:eq30}, it is impossible to hit the boundary at odd timesteps for even $x$ and at even timesteps for odd $x$. Thus, we only plotted the simulation results which has a nonzero probability to hit the boundary.}
	\label{fig:eq46}
\end{figure}

To find ${\mathrm{P}_{x^\prime}[M>N]}$, i.e., the probability that $\mathcal{B}$ stays in the $(-H,H)$ after $N$ time steps, we first note that \eqref{eq:hypergeom} include paths which allow $\mathcal{B}$ to traverse to $-\infty$ and return back. Since \eqref{eq:bertrand_exact} also suffers from the same limitation, as outlined in Fig. \ref{fig:transformed}, we use the same framework we outline in Sec. \ref{sec:reflect_cont} and reach
\begin{align}
\label{eq:pzf}
{P_{x^\prime}}[M>N] &= \sum_{i=0}^{m_1} (-1)^{i} \mathrm{P}^{(S)}_{x^\prime+2iH}[M>N]+\\  \nonumber&\phantom{oooooo}\sum_{i=0}^{m_2} (-1)^{i}\mathrm{P}^{(S)}_{2(i+1)H-x^\prime}[M>N],
\end{align}
where
\begin{align}
\label{eq:lim1}
m_1 & = \min (\{j|x^\prime+2jH>N\})-1,\\
\label{eq:lim2}
m_2 & = \min (\{j|2(j+1)H-x^\prime>N\})-1.
\end{align}

To calculate the rate of second incidence, we need to employ the difference equation on \eqref{eq:process_recursive}, as we derive \eqref{eq:binom_der} from \eqref{eq:bertrand_exact}. 
\begin{align}
\label{eq:px_rate}
\mathrm{P}_{{x^\prime}}[M=N] &=\mathrm{P}_{{x^\prime}}[M>N-2] - \mathrm{P}_{{x^\prime}}[M>N],
\end{align}
or alternatively,
\begin{align}
\label{eq:rate1}
R_2[x^\prime,N]=\frac{1}{2}\Big(\mathrm{P}_{{x^\prime}}[M>N-2] - \mathrm{P}_{{x^\prime}}[M>N]\Big).
\end{align}

Fig. \ref{fig:eq52} illustrates the accuracy of \eqref{eq:pzf}. As we state in the title of this section, we achieve our goal of making corrections to \eqref{eq:vartheta00} for small $N$.

\begin{figure}[t]
	\centering		\includegraphics[width=3.5in, trim={0cm 0 0cm 1.75cm}]{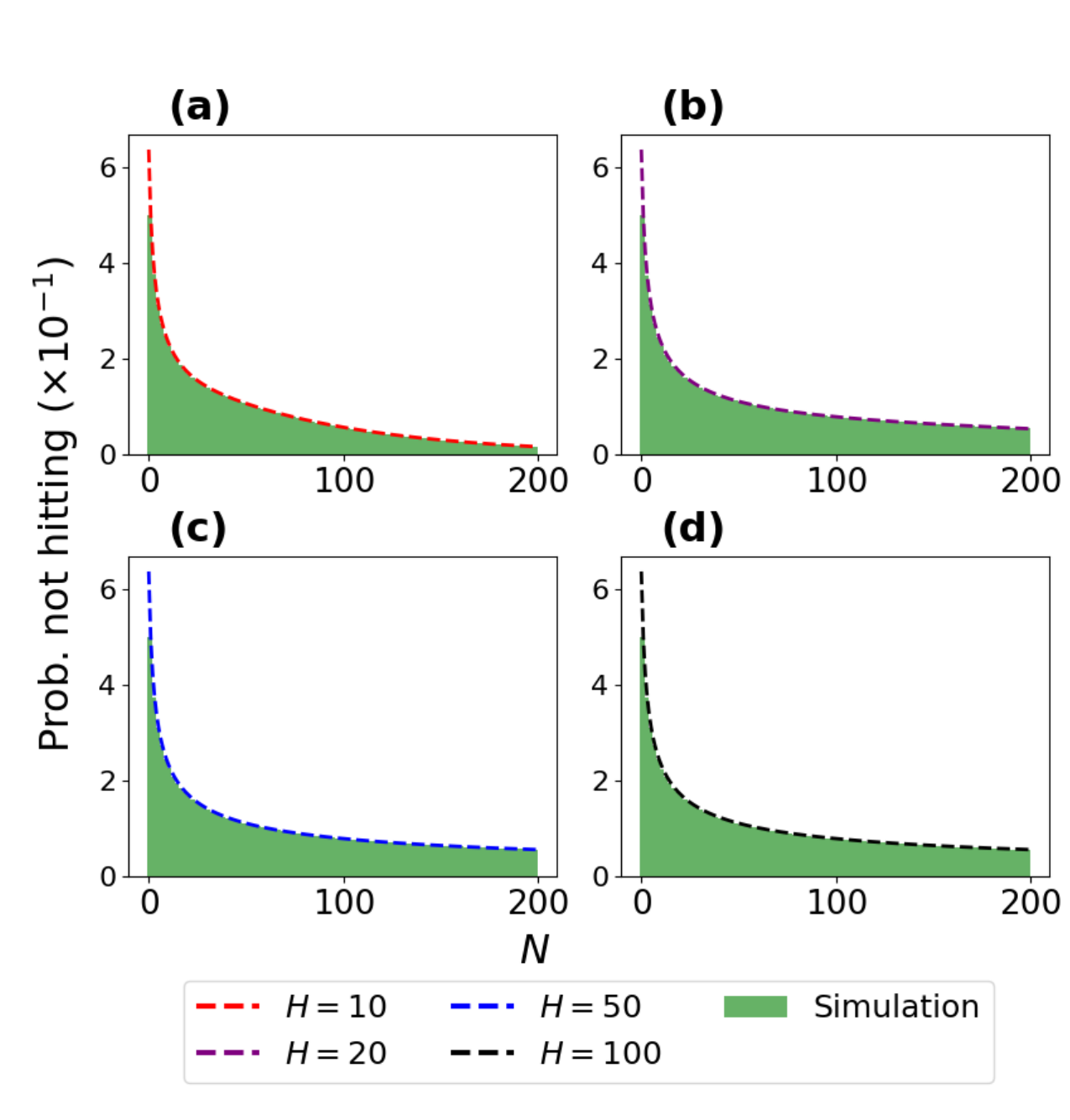}
	\vspace{-6mm}
	\caption{The probability of not hitting {\it any} boundary vs. number of timesteps assuming $\mathcal{B}$ starts one step away from the boundary. We use 1000000 iterations each to obtain the simulation results. The primary difference with the Fig. \ref{fig:eq46} is that here $\mathcal{B}$ can hit both $x=H$ and $x=-H$.}
	\label{fig:eq52}
\end{figure}

Note that, $i$, the summation index of \eqref{eq:pzf} represents the number of reflections. Accordingly, $i=0$ stands for no reflection and $i=1$ for the first summation represents reflection over $x=H$ and the second summation reflection over $x=-H$. The number of reflections in \eqref{eq:pzf} does not extend to infinity for finite $N$. Since we use \eqref{eq:vartheta00} for large $N$, summation limits are determined by the highest reflection number that requires less or equal steps than $N$. We determine the summation limits using \eqref{eq:lim1} and \eqref{eq:lim2}.

\section{Rate of Absorption in 1D}
\label{sec:rateabsicm}

In the previous sections, we derive expressions for the second incidence rate. In this section, we use these expressions to find the rate of removal of IcM from the medium.

We start with $R_h[N]$, the probability of higher order incidences on $\pm H$, for $\mathcal{B}$ starting at the edge of the domain. Using $R_2[N]$, i.e., the probability of second incidence,  obtained by either \eqref{eq:rate0} or \eqref{eq:rate1}, we reach the recursive relation
\begin{align}
\label{eq:incidence_ratef}
R_h[N] = R_2[N] + \sum_{i=1}^{N-1} R_2[i]R_2[N-1-i].
\end{align}

We can express \eqref{eq:incidence_ratef} as a difference equation, i.e.,
\begin{align}
\label{eq:incidence_ratef_diff}
R_h[N+1]-R_h &[N] = R_2[N+1] +\\
\nonumber
&\sum_{i=1}^{N}R_h[i]R_2[N-i+1].
\end{align}

\begin{figure}[!t]
	\centering	\includegraphics[width=3.5in]{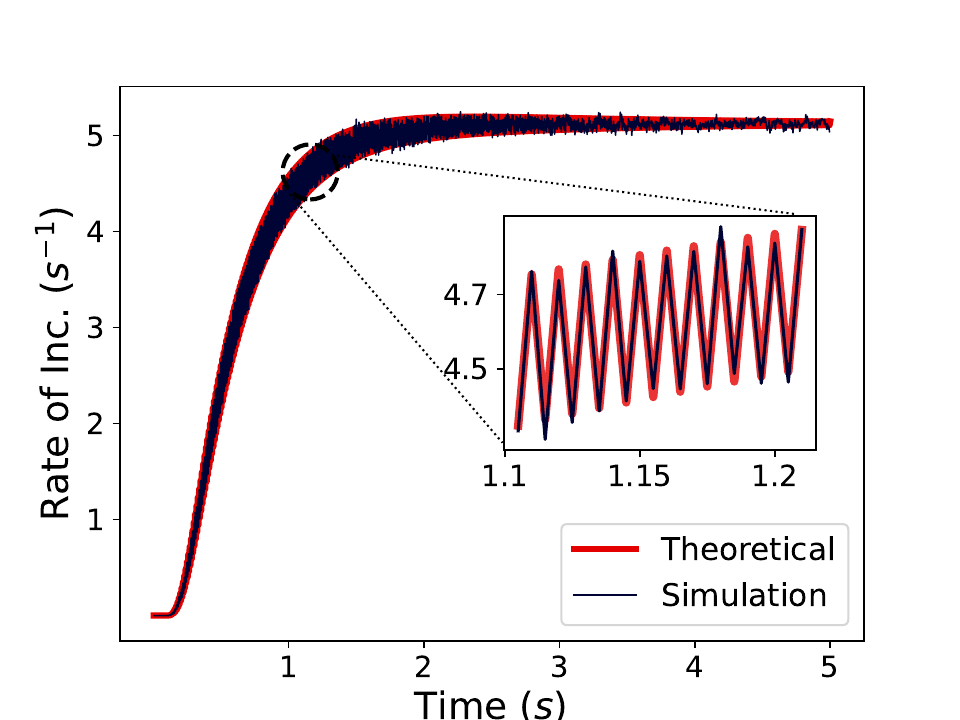}
	\vspace{-5mm}
	\caption{Incidence rate for a reflecting boundary with $\Delta{x} = 0.1$. Sample size is 2400000.}
	\label{fig:K1_s01}
\end{figure}

\begin{figure*}[!b]
	\centering	
	\includegraphics[width=7in]{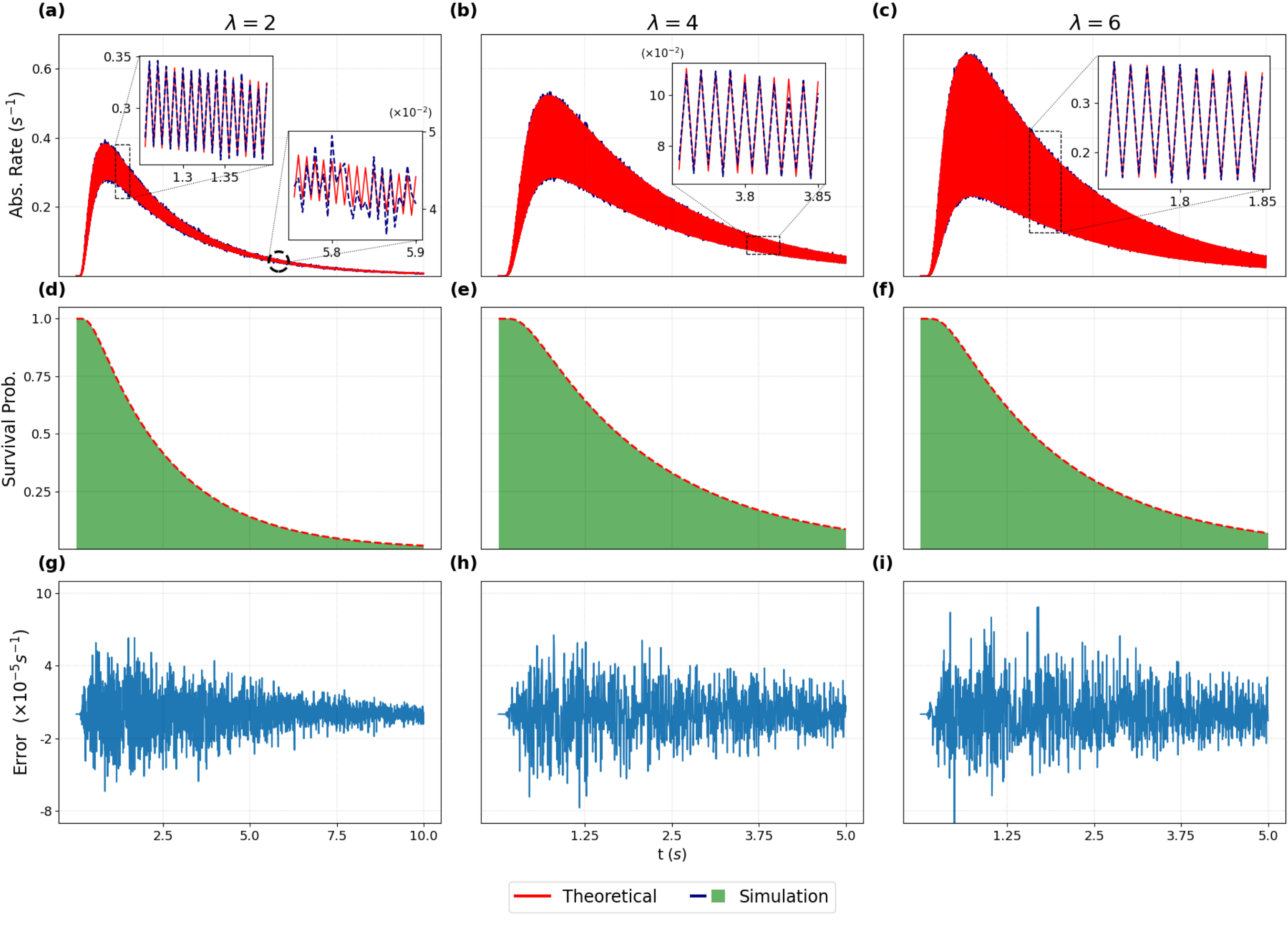}
	\vspace{-1mm}
	\caption{(a)-(c) Rate of absorption vs. time. (d)-(f) Survival probability vs. time. (g)-(i) Error in the rate of absorption vs. time. We choose $\Delta{x}=0.1$ and different $\lambda$ values. We take $t=10$ for $\lambda = 2$ to allow enough time for absorption.}
	\label{fig:big_sim}
\end{figure*}

In continuous time, \eqref{eq:incidence_ratef_diff} turns into a integro-differential equation of the form
\begin{align}
\label{eq:incidence_ratef_diff_conttime}
\frac{\mathrm{d}R_h}{\mathrm{d}t} = R_2(t)+\int_{0}^{t}R_h(t)R_2(\tau-t)\mathrm{d}\tau.
\end{align}

If needed, we can solve \eqref{eq:incidence_ratef} with z-transform, i.e.,
\begin{align}
\label{eq:incidence_ratef_diff_ztrans}
z\mathcal{R}_h(z)\!+\!z\!R_h[1]\!-\!\mathcal{R}_h(z)\! =\! \mathcal{R}_2(z)\! +\! z\!R_2[1]\! +\!\mathcal{R}_h(z)\mathcal{R}_2(z),
\end{align}
where $\mathcal{R}_h(z)=\mathcal{Z}\{R_h[N]\}$. Solving \eqref{eq:incidence_ratef_diff_ztrans}, we reach
\begin{align}
\label{eq:incidence_ratef_diff_ztrans_soln}
\mathcal{R}_h(z) = \frac{\mathcal{R}_2(z)}{z-\mathcal{R}_2(z)-1},
\end{align}
by using $R_h[1]=R_2[1]$, which is directly deducible from \eqref{eq:incidence_ratef_diff}. We can solve \eqref{eq:incidence_ratef_diff_conttime} similar to \eqref{eq:incidence_ratef_diff_ztrans_soln}, via Laplace transform, instead of z-transform.

Time complexity of \eqref{eq:incidence_ratef} is $\mathcal{O}(N^2)$, while time complexity of \eqref{eq:incidence_ratef_diff_ztrans_soln}, which relies on the z-transform is $\mathcal{O}(N\log(N))$. But, since array operations are faster than most other operations, it might be possible to find an approximate $R_h[N]$ through \eqref{eq:incidence_ratef} with the desired precision within reasonable time.

We can relax the edge start assumption by including the first incidence rate, $R_1$. We find $R_1[N]$ either by taking the derivative of the (11) of
\cite{koca2022channel} or by adjusting (29) of the same paper, i.e.,
\begin{align}
\nonumber
R_1[N] = \sum\limits_{i=-\infty}^\infty & (-1)^{i+1} \left[\frac{(2i+1)L}{\sqrt{8\pi N^3}}\exp\left(-\frac{(2i+1)^2L^2}{2N}\right)\right.\\
\label{eq:R_1}
&\left. -\frac{(2i-1)L}{\sqrt{8\pi N^3}}\exp\left(-\frac{(2i-1)^2L^2}{2N}\right)\right].
\end{align}

$R_1[N]$ provides the rate of incidence at time step $N$, when the motion starts at the center of the domain. However, as we state in Sec. \ref{sec:sys} and Sec. \ref{sec:reflect_cont}, \eqref{eq:R_1} is based on the Brownian Motion model described by \eqref{eq:alt_brown_def}. Although both models converge, for small $N$, \eqref{eq:R_1} allows incidence for $N\leq H$. To make \eqref{eq:R_1} compatible with \eqref{eq:incidence_ratef_diff}, we prohibit incidences for $N\leq H$, i.e., 
\begin{align}
R_1[N] = \begin{cases} 0 \phantom{o} &\mathtt{if } \phantom{o} N\leq H\\
R_1[N] \phantom{o} &\mathtt{if } \phantom{o} N>H
\end{cases}.
\end{align}

Finally, the rate of incidence for $\mathcal{B}[0]=0$ becomes
\begin{align}
\label{eq:fin_rate_of_inc}
R[N] = R_1[N] \ast R_h[N],
\end{align}
where $\ast$ stands for convolution. We demonstrate \eqref{eq:fin_rate_of_inc} along with the simulation results in Fig. \ref{fig:K1_s01}. Since there is no absorption in this setting, the incidence rate settles as concentration homogeneously disperses within the domain as expected.

To find the rate of absorption, $R_A$, we modify \eqref{eq:incidence_ratef_diff} to include $P_A$, i.e.,
\begin{align}
\label{eq:incidence_ratef_diff_abs}
R_h[N+1]-&R_h [N] = R_2[N+1] +\\
\nonumber
&(1-P_A)\sum_{i=1}^{N}R_h[i]R_2[N-i+1],
\end{align}
where $P_A$, the probability of absorption is
\begin{align}
\label{eq:erban_7}
P_A = \frac{\lambda}{D}\Delta{x},
\end{align}
directly from \cite{erban2007reactive}.

Finally, rate of absorption becomes
\begin{align}
\label{eq:abs_rate_amele}
R_A[N] = P_A \times R[N].
\end{align}

Fig. \ref{fig:big_sim} illustrates the accuracy of \eqref{eq:abs_rate_amele}. Apart from random noise factors, as can be seen in the zoomed in subplots in Fig. \ref{fig:big_sim}(a)-(c), our theoretical findings fit the Monte Carlo simulations perfectly. Similarly, as evident from Fig. \ref{fig:big_sim}(d)-(f), the theoretical survival probability and the simulation results are in perfect agreement. This perfect agreement is also visible in the error plots Fig. \ref{fig:big_sim}(g)-(i). Error is completely random and it thusly does not accumulate. Moreover, error magnitude diminishes as the incidence rate diminishes, which is another observation suggesting that errors are random rather than systematic.

\begin{figure}[!t]
	\centering	
	\includegraphics[width=3.5in]{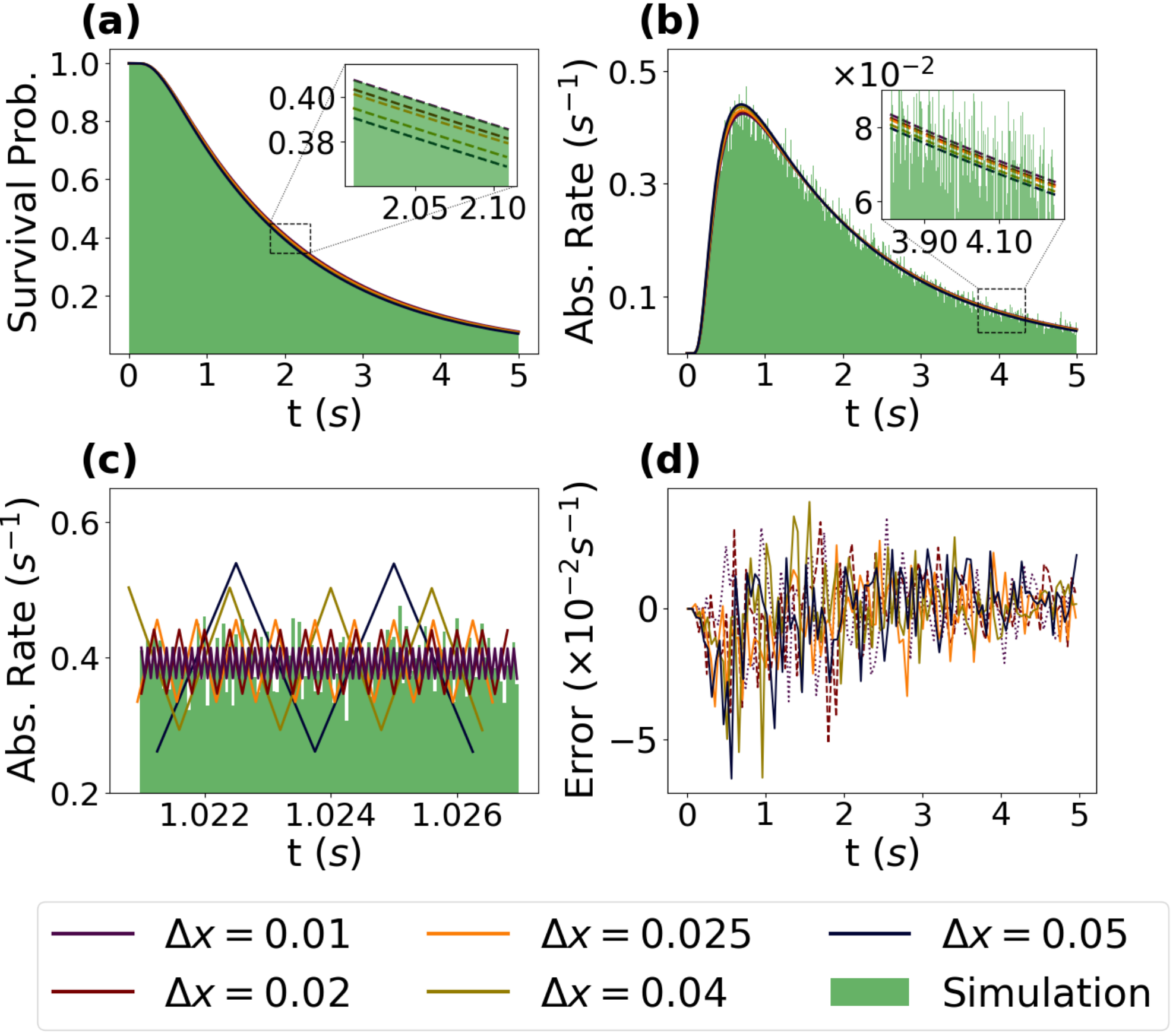}
	\vspace{-1mm}
	\caption{(a) Survival probability vs. time. (b) Smoothed rate of absorption vs. time. (c) Zoomed rate of absorption vs. time. (d) Error in the rate of absorption vs. time. We choose $\lambda=10$ and different $\Delta{x}$ values.}
	\label{fig:deltaxfig}
\end{figure}

Fig. \ref{fig:deltaxfig} demonstrates the impact of $\Delta{x}$ on the theoretical results. In Fig. \ref{fig:deltaxfig}(b), we smooth the rate with a moving average to eliminate the high frequency fluctuations arising from the parity, similar to the zigzags in Fig. \ref{fig:big_sim}. The raw rate is displayed in Fig. \ref{fig:deltaxfig}(c). To quantify the discrepancy between the simulation results and our theoretical findings, we use a reverse moving average approach to equate the lengths of the theoretical and simulation data. To this end, we first find $Q_{\Delta{t}}$, the ratio of simulation $\Delta{t}$ and the theoretical $\Delta{t}$, i.e.,
\begin{align}
Q_{\Delta{t}} = \frac{\Delta{x}^2/2D}{0.01^2/2D} = \left(\frac{\Delta{x}}{0.01}\right)^2.
\end{align}
When $Q_{\Delta{t}}$ is an integer, we can expand the theoretical result to the size of the simulations and find the reverse window averaged rate, $R_A^{(r)}[N]$, by
\begin{align}
R_A^{(r)}[N] = R_A[N] \otimes \underbrace{\left[\frac{1}{Q_{\Delta{t}}}, \dots, \frac{1}{Q_{\Delta{t}}}\right]}_{Q_{\Delta{t}} \phantom{o}\text{times}},
\end{align}
where $\otimes$ is the tensor product. When $Q_{\Delta{t}}$ is not an integer, we express $Q_{\Delta{t}}$ as a fraction, i.e., $Q_N/Q_D$. We then expand each $Q_D$ long slice of $R_A[N]$ to length $Q_N$ by both repetition and smoothing, satisfying
\begin{align}
\sum_{j \in j_N} R_A^{(r)}[j] = \sum_{i \in \{i_D\}}R_A[i], \phantom{ooo} 
\end{align}
where
\begin{align}
j_N &= \{kQ_N,\dots,k(Q_N+1)-1\},\\
i_D &= \{kQ_D,\dots,k(Q_D+1)-1\},\\
k &\in \left\lbrace 0, 1, \dots, \frac{t}{Q_D{\Delta{t}}}\right\rbrace.
\end{align}

Examining Fig. \ref{fig:deltaxfig}, we see that when we use a larger $\Delta{x}$ values than the simulation parameter, we can still very successfully estimate the probability of absorption. Furthermore, looking at Fig. \ref{fig:deltaxfig}(d), we observe that the error in the rate is relatively unaffected by $\Delta{x}$.  Thus, although \eqref{eq:incidence_ratef} is computationally heavy, by reducing $\Delta{x}$, we can reach the rate of absorption with the desired accuracy to reduce the computational burden. 

\section{Conclusion}
\label{sec:concpap01}

In this work, we derive expression for the second incidence rate, higher incidence rate and for the absorption rate in a Brownian Molecular Communication Channel. We verify each of our critical findings with Monte Carlo simulations.

Currently, our results are directly applicable to cuboid volumes with opposing partially absorbing boundaries. To address this shortcoming, we are already extending our work to three dimensional volumes with some or all boundaries being partially absorbing. The main limitation of our work is the $\mathcal{O}(N^2)$ complexity, limiting the scalability of our results. To remedy this, we plan to develop an approximate expression that is computed in $\mathcal{O}(N)$ and accurate for small absorption rates. Finally, we plan to apply our findings to Synaptic Molecular Communication in a future work.

\bibliographystyle{IEEEtranN}
\bibliography{krang} 

\appendices
\section*{Appendix A}
\begin{proof}[Proof of Theorem \ref{thm:recurse}]
We first define the set $\mathcal{S}(N)=\left\{\omega\in\{-1,+1\}^{N}\right\}$, where $-1$ denotes for a movement in $-x$ direction and $+1$ in $+x$ direction. We know that $\lvert \mathcal{S}(N) \rvert = 2^N$ where $\lvert \dots \rvert$ stands for the size of the set.

Let $\mathcal{S}(N,a) \subseteq \mathcal{S}(N)$, which consists of all series describing the trajectory of $\mathcal{B}$ starting at $x=H-a$ and never incident on $x=H$ within the first $N$ time steps. Then, the probability of $\mathcal{B}$, starting at $x=H-a$ and not incident on $x=H$ within the first $N$ time steps is
\begin{equation}
P_a(M>N) = \frac{\lvert \mathcal{S}(N,a) \rvert}{2^N}.
\end{equation}

Now, if $\mathcal{B}$ is not incident on either boundaries within the first $N$ steps, then $s$, the series that describe the motion of $\mathcal{B}$, obeys
\begin{equation}
s\in \mathcal{S}(N,1) \land s \in \mathcal{S}(N,2H-1),
\end{equation}
and
\begin{align}
\label{eq:p1_def}
P(s\in \mathcal{S}(N,1)) &= P_1(M>N),\\
\label{eq:p2h1_def}
P(s\in \mathcal{S}(N,2H-1)) &= P_{2H-1}(M>N).
\end{align}

Note that we have already calculated \eqref{eq:p1_def} with \eqref{eq:bertrand_exact}. However, we cannot use the Bertrand's Ballot Theorem directly to find \eqref{eq:p2h1_def}. To remedy this, we deduce a recursive relation to find $\lvert \mathcal{S}(N,2H-1) \rvert$.

Since all $s \in \mathcal{S}(N,1)$ has to start with a $-1$, we can find $\mathcal{S}(N,2)$ by removing the first element of all $s \in \mathcal{S}(N+1,1)$, i.e.,
\begin{equation}
\label{eq:process_4}
\mathcal{S}(N,2) = \{s[1:]\hspace{0mm}\mid s \in \mathcal{S}(N+1,1)\},
\end{equation}
where we use the python array slicing convention. Following \eqref{eq:process_4}, $P_2(M>N)$ becomes
	
\begin{align}
\nonumber
P_2(M>N) &= 2^{-N}\lvert \mathcal{S}(N,2) \rvert \\
\nonumber
&= 2 ^ {-N} \lvert \mathcal{S}(N+1,1) \rvert \\
\label{eq:process_5}
&= 2P_1(M>N+1).
\end{align}

Similarly,
\begin{align}
\nonumber
\lvert \mathcal{S}(N,3)\rvert &= \lvert \mathcal{S}(N+1,2) \rvert - \lvert \mathcal{S}(N,1) \rvert\\
&= \lvert\mathcal{S}(N+2,1)\rvert - \lvert\mathcal{S}(N,1)\rvert,
\end{align}
thus we can calculate $P_3(M>N)$ as
\begin{equation}
\label{eq:process_6}
P_3(M>N) = 4P_1(M>N+2)- P_1(M>N).
\end{equation}

It is evident that we can calculate $P_x(M>N)$ in a recursive fashion, i.e.,
\begin{align}
P_x(M>N) \!&=\! 2P_{x-1}(M>N+1)\! -\! P_{x-2}(M>N).
\end{align}
\end{proof}

\section*{Appendix B}
\begin{proof}[Proof of Theorem \ref{thm:hypergeom}]
We notice that \eqref{eq:process_recursive}, the recursion relation provided of the Extended Ballot Theorem by Theorem \ref{thm:recurse}, is the same recursion relation given by the Chebyshev polynomials of the second kind. Accordingly, we find the direct expression for $P_x(M>N)$, that does not necessitate recursive calculations as

\begin{align}
\label{eq:cheb0}
P_x(M>N) &= \sum_{i=0}^{\xi} (-1)^i 2^{x-1-2i}\binom{x-i-1}{i} \times \\ 
\nonumber
&\phantom{oooooooo}P_1(M>N+x-2i+\xi^\prime)
\end{align}

where $\xi$, $\xi^\prime$ and $\xi^{\prime\prime}$ are defined by (\ref{eq:def_xi1}-\ref{eq:def_xipprime}). As expected, the coefficients of $P_1$ are the same coefficients as the $\alpha^{x-1-2i}$ term of the Chebyshev polynomials of the second kind \cite{oeis}.

By substituting \eqref{eq:bertrand_exact} for $P_1$ and further manipulating  \eqref{eq:cheb0}, we reach
\begin{align}
\label{eq:hyper_01}
P_x(M>N) &= \sum_{i=0}^q (-1)^i\binom{2q+\xi^\prime-i}{i}2^{2q+\xi^\prime-2i}\times 
\nonumber\\ &\phantom{ooo}\frac{1}{2^{N+2q-2i+\xi^{\prime\prime}}}\binom{N+2q-2i+\xi^{\prime\prime}}{\frac{N+\xi^{\prime\prime}}{2}+q-i}.
\end{align}

To avoid the falling factorials we make the $i \leftarrow q-i$ substitution to the \eqref{eq:hyper_01} and reach
\begin{align}
\label{eq:hyper_02}
P_x(M>N) &= 2^{\xi^\prime}(-1)^{\xi}\sum_{i=0}^q (-1)^i\binom{q+i+\xi^\prime}{q-i}4^{i}\times 
\nonumber\\ &\phantom{ooo}\frac{1}{2^{N+2i+\xi^{\prime\prime}}}\binom{N+2i+\xi^{\prime\prime}}{\frac{N+\xi^{\prime\prime}}{2}+i}.
\end{align}

We apply Legendre's Duplication Formula of the form
\begin{align}
\label{eq:legendres}
\frac{1}{2^{2z}}\binom{2z}{z}=\frac{\Gamma\left(z+\frac{1}{2}\right)}{\sqrt{\pi}\Gamma\left(z+1\right)},
\end{align}
to \eqref{eq:hyper_02} and obtain
\begin{align}
\label{eq:hyper_03}
P_x(M>N) &= \frac{2^{\xi^\prime}(-1)^{\xi}}{\sqrt{\pi}}\sum_{i=0}^q (-1)^i\binom{q+i+\xi^\prime}{q-i}\times \\ 
\nonumber &\phantom{ooooo}\frac{\Gamma\left(\Delta + i +\frac{1}{2}\right)}{\Gamma\left(\Delta + i + 1\right)}4^{i}.
\end{align}

Finally we perform arithmetic manipulations to transform \eqref{eq:hyper_03} into \eqref{eq:hypergeom_def} format, i.e.,
\begin{align}
P_x(M>N)&=\frac{2^{\xi^\prime}(-1)^{\xi}(q+1)^{\xi^\prime}}{\sqrt{\pi}}\frac{\Gamma\left(\Delta +\frac{1}{2}\right)}{\Gamma\left(\Delta + 1\right)}\times \\ 
\nonumber &\sum_{i=0}^q (-1)^i\binom{q}{i}\frac{(q+\xi^\prime+1)_i}{\left(\frac{1}{2}+\xi^\prime\right)_i} \frac{\left(\Delta + \frac{1}{2}\right)_i}{\left(\Delta + 1\right)_i},
\end{align}
which is the expansion of \eqref{eq:hypergeom}.\end{proof}

\end{document}